\documentclass[11pt, article]{amsart}
\usepackage{amsfonts}
\usepackage{amssymb}
\usepackage[dvips]{graphics}
\usepackage{wrapfig}
\usepackage{subcaption} 
\usepackage{epsfig}
\usepackage{color}
\usepackage[pdftex]{hyperref}
\usepackage{fullpage}
\usepackage{preamble}

\newcommand{\curv}{\kappa}

\newlength{\imagewidth}
\begin{document}

\title{Quantization of conductance in gapped interacting systems}

\author{Sven Bachmann}
\address{Department of Mathematics \\ University of British Columbia \\ Vancouver, BC V6T 1Z2 \\ Canada}
\email{sbach@math.ubc.ca}

\author{Alex Bols}
\address{ Instituut Theoretische Fysica, KU Leuven  \\
3001 Leuven  \\ Belgium }
\email{alexander.bols@kuleuven.be}

\author{Wojciech de Roeck}
\address{ Instituut Theoretische Fysica, KU Leuven  \\
3001 Leuven  \\ Belgium }
\email{wojciech.deroeck@kuleuven.be}

\author{Martin Fraas}
\address{ Instituut Theoretische Fysica, KU Leuven  \\
3001 Leuven  \\ Belgium }
\curraddr{Department of Mathematics \\ Virginia Tech \\ Blacksburg, VA 24061-0123 \\ USA}
\email{fraas@vt.edu}

\date{\today }

\begin{abstract}   We provide a short  proof of the quantisation of the Hall conductance for gapped interacting quantum lattice systems on the two-dimensional torus. This is not new and should be seen as an adaptation of the proof of~\cite{HastingsMichalakis}, simplified by making the stronger assumption that the Hamiltonian remains gapped when threading the torus with fluxes. We argue why this assumption is very plausible. The conductance is given by Berry's curvature and our key auxiliary result is that the curvature is asymptotically constant across the torus of fluxes.
\end{abstract}

\maketitle

\section{Setup and Results}\label{sec:Results}

\subsection{Preamble}
It is now common lore that the remarkable precision of the plateaus appearing in Hall measurements at low temperatures is explained by linking the Hall conductance with a topological invariant. For translationally invariant, non-interacting systems, it is the Chern number of  the ground state bundle over the Brillouin zone.  In interacting systems, the Brillouin zone is replaced by a torus associated with fluxes threading the system. {In independent works,   Avron and Seiler \cite{AvronSeiler85} and Thouless, Niu and Wu \cite{Thouless85}  prove quantisation of the Hall conductance in this framework \emph{assuming} that the
adiabatic curvature of the ground state bundle is constant, i.e.\ independent of the fluxes\footnote{Thouless and Niu argue in~\cite{Niu87} why the assumption is reasonable, relying on locality arguments that foreshadow the later proof. The assumption can be replaced by averaging the conductance over the flux torus. In a slightly different setting~\cite{Laughlin}, Laughlin argues that the averaging over \emph{one of two} fluxes can actually be justified.}. 

Proving the constancy of curvature, or bypassing it, was considered an open problem \cite{Avron:web}, and was resolved only thirty years later by Hastings and Michalakis in~\cite{HastingsMichalakis} by relying on a crucial locality estimate.  
 
 The present paper gives a streamlined and expository version of the proof in~\cite{HastingsMichalakis}, presenting also a result in the thermodynamic limit. Our version is shorter, at the cost of making a stronger assumption. Indeed, we assume that the gap remains open for the system threaded with fluxes. This is a prerequisite to even speak about the adiabatic curvature on the torus of fluxes, cf.\ the framework discussed above. Remarkably,  \cite{HastingsMichalakis} don't need this assumption as they bypass the use of bundles. 
 
  A recent work of Giuliani, Mastropietro and Porta~\cite{Giuliani:2016gn} yields a similar result, namely the quantization of the Hall conductance for interacting electrons in the thermodynamic limit, restricted to weak interactions. They also bypass the geometric picture in favour of Ward identities and constructive quantum field theory. Finally, we note that the quantization is also well understood via effective field theories \cite{frohlich1995quantum} (in casu: Chern-Simons).

\subsection{Quantum lattice systems}\label{sub:Setup}

We consider a two-dimensional discrete torus $\Gamma$ with $L^2$ sites, which we identify with a square $[0,L) \times [0,L)\cap\bbZ^2$ whose edges are glued together. For simplicity we assume that $L$ is even. A finite-dimensional Hilbert space $\mathcal{H}_x$ is associated to each site of the torus and for a subset $X$ of the torus we define $\mathcal{H}_X = \otimes_{x \in X} \mathcal{H}_x$.  The evolution of the system is governed by a finite range Hamiltonian
\begin{equation}\label{Ham}
H = \sum_{X \subset \Gamma} \Phi(X), \quad \Phi(X) \in \mathcal B(\mathcal{H}_X)
\end{equation}
that is assumed to be gapped, see below. By finite range, we mean that
\begin{equation*}
 \Phi(X)  = 0\quad\text{whenever}\quad \mathrm{diam}(X)>R.
\end{equation*}
As usual, we identify operators acting on a subset $X$ with their trivial extension to $\Gamma$ by
\begin{equation}\label{embedding}
A\in\mathcal B(\caH_X) \quad\longleftrightarrow\quad A\otimes \mathbb I_{\Gamma\setminus X}\in\mathcal B(\caH_\Gamma).
\end{equation}

We are interested in charge transport. The charge at site $x$ is given by a Hermitian operator $Q_x$ that takes integer values, namely its spectrum is a finite subset of $\bbZ$.  The total charge in a region $X$ is then given by
$$
Q_X = \sum_{x \in X} Q_x.
$$
The charge is a locally conserved quantity:
\begin{equation}\label{eq: charge conservation}
[Q_X, \Phi(Y)]=0, \quad \text{if} \quad Y\subset X \subset\Gamma.
\end{equation}
As we will often have to deal with boundaries of spatial regions, we introduce the following sets
\begin{equation*}
X^r = \{x \in \Gamma : \mathrm{dist}(x, X) \leq r\},\qquad X_r = \{x \in \Gamma : \mathrm{dist}(x, \Gamma\setminus X) \leq r\}
\end{equation*}
and
\begin{equation*}
\partial X(r) = X^r\cap X_r,
\end{equation*}
which corresponds to symmetric ribbon of width $2r$ around the boundary of $X$.

We shall denote $\partial X \equiv \partial X(R)$ since this is practically the only case of relevance. In particular, it follows from charge conservation and the fact that $\Phi$ has finite range $R$ that $[Q_X, H] \in \mathcal{B}(\caH_{\partial X})$.

For any $X \subset \Gamma$ and any operator $A$ we write 
\begin{equation*}
\tr_{X} (A) := \frac{1}{\dim \caH_X} \Tr_X(A)
\end{equation*}
for the normalized partial trace $\tr_X: \caB(\caH_\Gamma)\to\caB(\caH_{\Gamma\setminus X})$ with respect to the set $X$.

\noindent\textbf{Remark.}
Whereas the above setting is phrased in terms of a quantum spin system and on rectangular lattice, this is not necessary. One can equally well consider fermions on the lattice and other types of lattices.

In the fermionic picture, the algebras of observables $\mathcal{B}(\caH_X)$ are replaced by the algebra $\caA_X$ of canonical anticommutation relations built upon $l^2(X; \bbC^N)$. The anticommutation properties of fermionic observables require one further restriction and one change to keep the crucial locality properties of a quantum spin system. First of all, the interactions $\Phi(X)$ and charges $Q_x$ must be even in the fermionic creation/annihilation operators. Secondly, the partial trace $\tr_{\Gamma\setminus X}$ must be replaced by another projection $\mathbb{E}_X: \caA_\Gamma\to\caA_X$. See Section~4 in~\cite{BrunoFermions} for details. With this, the Lieb-Robinson bound and its corollaries carry over to lattice fermion systems, see~\cite{BrunoFermions, Bru:2017aa}.

The advantage of this extension is that there are natural examples that fit our scheme, most notably the (second quantized) Haldane and Harper models with a small interaction term added to them, see~\cite{HastingsPerturbation, De-Roeck:aa}.

\subsection{Hamiltonians with fluxes}\label{sec:fluxes}
We consider regions $X_1 = \{ (x_1,x_2) \in\Gamma : 0 \leq x_2 \leq L/2\}$ resp. $X_2 = \{ (x_1,x_2) \in\Gamma : 0 \leq x_1 \leq L/2\}$ and the associated charges $Q_j = Q_{X_j}$.

By charge conservation, $[Q_i, H]$ is supported on $\partial X_i$. If $L/2 > R$ (which we shall assume from now on) then $\partial X_1$ consists of two disjoint ribbons of width $2R$ and centered around the lines $x_2 = 0$ and $x_2 = L/2$. We will denote these ribbons by $\partial X_1^-$ and $\partial X_1^+$ respectively, and introduce the analogous sets for $X_2$. Finally, we let
\begin{equation*}
\Delta = \partial X_1^- \cap \partial X_2^-,
\end{equation*}
see Figure~\ref{fig: four squares}. 

We now define two one-parameter groups of unitaries by
\begin{equation*}
U_j(\varphi) := e^{-\iu \varphi Q_j},\qquad j=1,2.
\end{equation*}
The integrality of the spectrum of $Q_j$ implies that $\varphi\mapsto U_j(\varphi)$ are periodic with period $2\pi$. With this, the \emph{flux Hamiltonians}, which depend on four angles $(\phi_-,\phi_+)= ((\phi^1_-,\phi_-^2),(\phi^1_+,\phi_+^2))\in\bbT^2\times\bbT^2$ (where $\bbT^2 = [0,2\pi]\times[0,2\pi]$ is the $2$-torus), are defined by
\begin{equation*}
H(\phi_-,\phi_+) =\sum_{Y\subset\Gamma} 
 U_2(\phi^2_Y)\str  U_1(\phi^1_Y)\str  \Phi(Y) U_1(\phi^1_Y)  U_2(\phi^2_Y)
\end{equation*}
where
\begin{equation*}
\phi^j_Y = 
\begin{cases}
\phi^j_- & \text{if }Y \cap \partial X_j^- \neq\emptyset \\
\phi^j_+ & \text{if }Y \cap \partial X_j^+ \neq\emptyset \\
0 & \text{if }Y \cap \partial X_j = \emptyset
\end{cases}
\end{equation*}

Note, first of all, that the order of the `twisting', which we take above first across horizontal lines and then across vertical lines is irrelevant as $[Q_1,Q_2] = 0$. Second of all, $H(\phi_-,\phi_+)$ are not all unitarily equivalent to each other. However, for any $\theta\in\bbT^2$, we have that
\begin{equation}\label{GT}
H(\phi_-+\theta,\phi_+-\theta) = (U_1(\theta^1)U_2(\theta^2))\str H(\phi_-,\phi_+) (U_1(\theta^1)U_2(\theta^2)),
\end{equation}
by charge conservation~(\ref{eq: charge conservation}).

Although these general Hamiltonians are briefly needed in the proofs, the key players will be on the one hand the \emph{twist} Hamiltonians
\begin{equation}\label{HT}
\tilde H (\phi) = H(\phi,0),\qquad\phi\in\bbT^2,
\end{equation}
and the \emph{twist-antitwist} Hamiltonians
\begin{equation}\label{HTAT}
H(\phi) = H(\phi,-\phi),\qquad\phi\in\bbT^2,
\end{equation}
Clearly, $\tilde H(0) = H(0) = H$. Moreover, the twist-antitwist Hamiltonians are all unitarily equivalent to each other, but this does not hold for the twist Hamiltonians. The physical picture for the twist Hamiltonian is that
a flux pair $\phi=(\phi_1,\phi_2) \in \bbT^2$ is threaded through the torus along the $x_1,x_2$ axes. We also point out that the fluxes discussed here are  fluxes on top of those possibly contained in $H$, hence not necessarily the total physical fluxes. In particular, $H$ also contains the magnetic field piercing the torus necessary to have a possibly non-zero Hall conductivity.

Finally, we note that by construction,
\begin{equation*}
\mathrm{supp}(H(\phi_-,\phi_+) - H) \subset \partial X_1\cup\partial X_2,
\end{equation*}
as well as
\begin{equation}\label{supp T vs TAT}
\mathrm{supp}(\tilde H(\phi) - H) \subset \partial X_1^-\cup\partial X_2^-,\qquad
\mathrm{supp}(\tilde H(\phi) - H(\phi)) \subset \partial X_1^+\cup\partial X_2^+,
\end{equation}
see Figure~\ref{fig: four squares}.

\begin{figure}
\centering
\captionsetup[subfigure]{width=\imagewidth}%
  \begin{subfigure}[b]{0.49\textwidth}
  	\centering
      \includegraphics[width=0.75\textwidth]{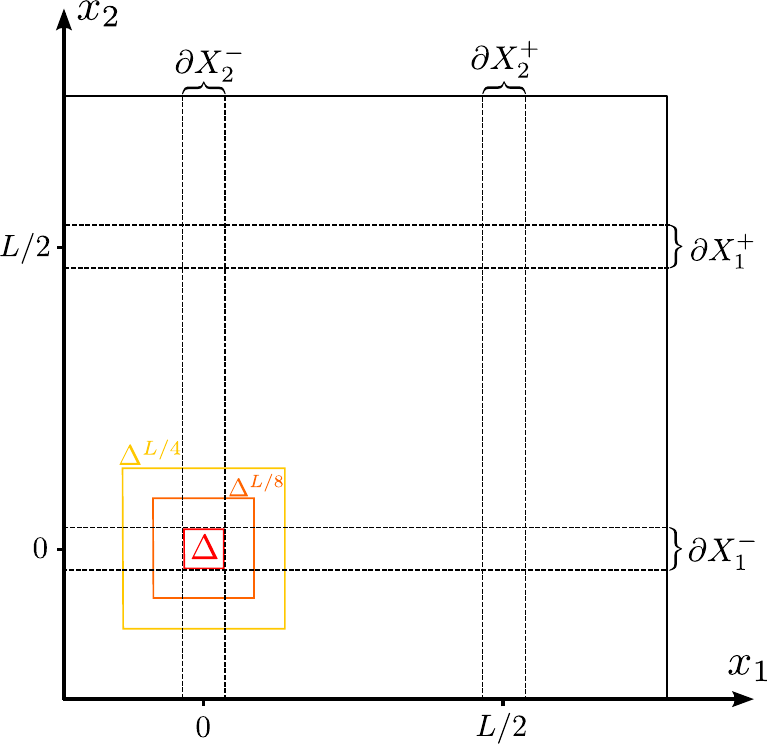}
      \caption{The spatial regions relevant for defining $H(\phi_-,\phi_+)$.}
      \label{fig: four squares}
  \end{subfigure}
  \hfill
  \begin{subfigure}[b]{0.49\textwidth}
  	\centering
      \includegraphics[width=0.75\textwidth]{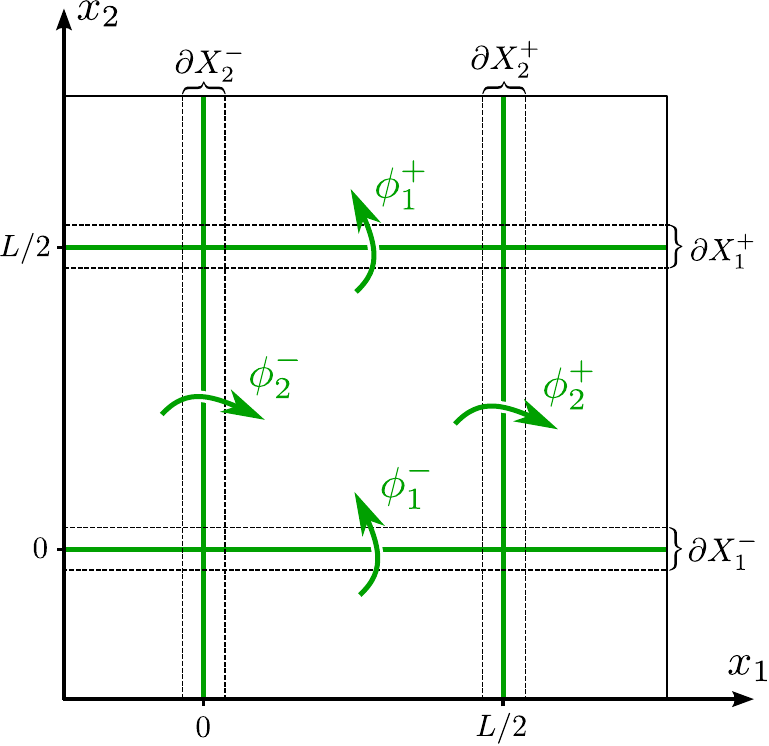}
      \caption{The fluxes are threaded across the green lines.}
      \label{fig:twists}
  \end{subfigure}
  \caption{The adiabatic curvature $[\partial_1 \tilde{P}(\phi), \partial_2 \tilde{P}(\phi)]$ associated with the twist Hamiltonians, for which $\phi_1^+ = \phi_2^+ = 0$, is supported in a neighbourhood of $\Delta$.}
\end{figure}

\begin{assumption}[Gap for all $\phi$]\label{tilde gap}
$\tilde H(\phi)$ has a non-degenerate ground state whose distance to the rest of the spectrum, i.e.\ the \emph{gap}, is bounded below by $\gamma(\phi) > 0$, uniformly in $L$. 
Moreover, $\inf_{{(\phi)\in\bbT^2}} \gamma(\phi) \geq \gamma$ for some $\gamma>0$. 
\end{assumption}
}
This assumption will be in place throughout the entire paper, so we do not repeat it. By gauge covariance (\ref{GT}), all flux Hamiltonians are gapped. Making the assumption is standard in the context of the quantum Hall effect.  Hastings and Michalakis \cite{HastingsMichalakis} assume the gap condition only at one point $\phi = \phi_0$ of the flux torus. In Section~\ref{sec: rationale} we explain why one could believe this assumption to hold true for any $\phi$ if it holds for $\phi = \phi_0$.

\subsection{Results}
We denote by $\tilde{P}(\phi)$ the ground state projection of $\tilde{H}(\phi)$ and, 
for the sake of recognizability, we write
$$
\omega_{\phi}(O)=\omega_{\phi,L}(O) =\Tr (\tilde P(\phi) O)
$$
for the ground state expectations.    The Hall adiabatic curvature is defined by
$$
\curv(\phi) = \iu \omega_{\phi} ([\partial_1 \tilde{P}(\phi), \partial_2 \tilde{P}(\phi)]),
$$
where we denoted $\partial_j = \partial / \partial \phi_j$. The main point proven in this note is 
\begin{prop}\label{prop: constant curvature}
 The Hall adiabatic curvature is asymptotically $\phi$-independent, in that, for any $N>0$,
$$ 
\sup_{\phi,\phi'\in\bbT^2} |\curv(\phi)-\curv(\phi')|  \leq C(N)L^{-N},
$$
where $C(N)$ is independent of $L$.
\end{prop}

Since the integral of curvature is an integer multiple of $2\pi$, see e.g.~\cite{AvronSeiler85}, this immediately implies
\begin{thm}\label{thm:main}
For any $\phi \in \bbT^2$ and any $N>0$
$$
\inf_{n\in \bbZ} | \curv(\phi)-2\pi n|  \leq C(N)L^{-N}.
$$
Moreover, the minimizer $n_0$ is independent of $\phi$.
\end{thm}
It  is common lore that $\curv(0)$ is the Hall conductance of the original model described by $H$, see~\cite{TKNN}.
The arguments used up to now do not give any information on how $\curv(\phi)$ depends on $L$, and indeed the integer $n_0$ may a priori depend on $L$. To clarify this, it is natural to assume that the state $\omega_0$  has a thermodynamic limit: 
\begin{thm}\label{thm:mainupgrade}
Assume that for any operator $O \in \mathcal B(\mathcal H_X)$ with $X$ finite, the limit
$
\lim_{L\to \infty} \omega_{0,L}(O)
$
exists. Then, the thermodynamic limit of the Hall adiabatic curvature exists and it is quantized:
$$
\lim_{L\to \infty} \curv(0)    \in  2\pi \bbZ.
$$
\end{thm}

Recent works \cite{BDF16,BDF17,Teufel17} provided a proof that the Hall conductance equals the Hall adiabatic curvature (also) in interacting systems.


\subsection{The rationale for Assumption \ref{tilde gap}} \label{sec: rationale}

Consider, for a function $\alpha: \Gamma \to \bbR$, the unitary (gauge transformation) $U_\alpha=\prod_{x \in \Gamma} e^{-\mathrm i\alpha(x)Q_x }$ and choose, for given $\phi=(\phi_1,\phi_2)$ 
$$
\alpha_{\phi}(x_1,x_2)= {-(1-x_1/L){\phi_1} - (1-x_2/L){\phi_2}}.
$$
Then we check that 
\begin{equation} \label{eq: distributing flux}
U_{\alpha_\phi} \tilde H(\phi) U^*_{\alpha_\phi}=H+ W(\phi)
\end{equation}
where $W=W(\phi)$ is of the form $W=\sum_{X\subset\Gamma} W(X)$ with $W(X) \in \caB(\caH_X)$ and such that
\begin{enumerate}
\item  $W(X)=0$ whenever $\mathrm{diam}(X)>R$,
\item  $\sup_{X\subset\Gamma}    \Vert  W(X)  \Vert  \leq \epsilon $.
\end{enumerate}
In fact, we have here that $\epsilon = C/L$ for some $L$-independent constant $C$. Although this can be checked by a direct calculation, it is best understood as follows. First of all, local charge conservation~\eqref{eq: charge conservation} implies that the effect of a $U_\alpha$ on a local interaction term, say $\Phi(X)$, depends only on the \emph{change} of $\alpha(x)$ over $X$. In the proposed $\alpha_\phi$, this is of order $L^{-1}$ everywhere but across the site $L-1$ and $0$. There however, this abrupt jump of size $-\phi_1$ is precisely compensated by the twist induced by $U(\phi)$ in $\tilde H(\phi)$. Put differently, a twist-antitwist can be removed by a gauge transformation using a vector potential that is a single-valued function on the torus. A twist cannot be removed globally as it corresponds to a multivalued vector potential, but as such its effect can still be made locally small everywhere.

The stability of the spectral gap for a Hamiltonian can be formulated as follows. 
A Hamiltonian $H$ with a non-degenerate ground state has a stable spectral gap, if for any $W$  satisfying conditions $\mathrm{i,ii,}$ with $\epsilon$ sufficiently small but $L$-independent,
$H+W$ has a non-degenerate ground state with a gap, uniformly in $L$.  At the time of writing, stability of the spectral gap has been proven in the case $H$ is \emph{frustration-free}~\cite{Bravyi:2011ea, Michalakis:2013gh} or the second quantization of free fermions~\cite{HastingsPerturbation, De-Roeck:aa}. The latter case being, arguably, the most relevant for quantum Hall effect. Yet, \emph{if} $H = H(0)$ has a stable gap in the precise sense above, then, by \eqref{eq: distributing flux},  Assumption~\ref{tilde gap} holds true as well, i.e.\ all $\tilde H(\phi)$ are gapped uniformly in $\phi$.

On the other hand, counterexamples of Hamiltonians with an unstable gap were constructed \cite{Michalakis:2013gh} or proposed specifically for our setting~\cite{Hastingsprivate}. Unlike our result, \cite{HastingsMichalakis} also covers those cases because the gap assumption there is only made for $\phi=0$. Therefore, the authors of~\cite{HastingsMichalakis} need a vastly more ingenious proof than we do. However, the observation~(\ref{eq: distributing flux}) and the fact that stability holds true for free fermions make us believe that Assumption~\ref{tilde gap} for all $\phi$ is reasonable from the physical point of view.

\section{Preliminaries}

\noindent We recall some standard results on locality of the dynamics of quantum lattice systems that will be crucial for our proofs. 
\subsection{Lieb-Robinson bounds and consequences}

As the flux Hamiltonians $H(\phi_-,\phi_+)$ are sums of local terms, they all satisfy a Lieb-Robinson bound~\cite{Lieb:1972ts,nachtergaele2010lieb}.
\begin{lemma}\label{lemma: lr}
There exists constants $v, C > 0$ such that for any $O_X\in \caB(\caH_X),O_Y\in \caB(\caH_Y)$, 
\begin{equation*}
\Vert [\tau_t(O_X), O_Y] \Vert \leq C \Vert O_X \Vert \Vert O_Y \Vert \min \left( \abs{X}, \abs{Y} \right) \ep{-\left( \dist(X, Y) - v \abs{t} \right)}
\end{equation*}
for all $t\in\bbR$, where $\tau_t(\cdot)$ is the dynamics generated by any flux Hamiltonian.
\end{lemma}
Note that bound is valid for all flux Hamiltonians and all system sizes. In particular, the constants $v,C$ are chosen to be independent of both $(\phi_-,\phi_+)$ and of $L$.

Here are two direct consequences of the Lieb-Robinson bound.
\begin{lemma}\label{lma:Diff of H}
Let $H, H'$ be two flux Hamiltonians and let $\tau_t^H,\tau_t^{H'}$ be the corresponding dynamics. Then for any $O_X\in\caB(\caH_X)$, 
\begin{equation*}
\left \Vert \tau_t^H (O_X) - \tau_t^{H'}(O_X)\right\Vert \leq C \Vert O_X \Vert \vert X\vert L \ep{-(\mathrm{dist}(\partial X_1 \cup \partial X_2,X)-v\vert t\vert)}.
\end{equation*}
\end{lemma}
\begin{proof}
Starting from
\begin{equation*}
\tau_{-t}^H \tau_{t}^{H'}(O_X) - O_X = -\iu\int_0^t\tau_{-s}^H\left([H - H',\tau_s^{H'}(O_X)]\right)\dif s,
\end{equation*}
the bound follows by the Lieb-Robinson bound for $\tau^{H'}$, unitarity of the evolution $\tau^H$, and the fact that $\mathrm{supp}(H-H')\subset \partial X_1 \cup \partial X_2$ whose volume is proportional to $L$.
\end{proof}

The second consequence of the Lieb-Robinson bound is then
\begin{lemma}\label{lma:trace of evolved}
Let $H$ be a flux Hamiltonian.  Then for any $O_X\in\caB(\caH_X)$, 
\begin{equation*}
\left \Vert  \tau_t^H(O_X)-\tr_{\Gamma\setminus X^r} \tau_t^H(O_X)\right\Vert \leq C \Vert O_X\Vert\vert X \vert\ep{-(r-v\vert t\vert)},
\end{equation*}
where $X^r = \{x\in\Gamma: \mathrm{dist}(x,X)\leq r\}$.
\end{lemma}
\noindent For a proof, we refer e.g.~to~\cite{LRYoshiko}. {Note that the norm of the difference is well-defined by the identification~(\ref{embedding}).}

\subsection{Quasi-adiabatic evolution}

For the following result, we refer to~\cite{HastingsWen,Sven}, whose results apply in this context by our gap Assumption~\ref{tilde gap}.

\begin{lemma}\label{lem: quasi adiabatic}
Let $s \mapsto (\phi_-(s),\phi_+(s)) \in \bbT^2\times \bbT^2$ for $s \in [0, 1]$ be a differentiable curve of fluxes, and denote by $H(s)=H(\phi_-(s),\phi_+(s))$
the corresponding family of Hamiltonians. Then there is a family of unitaries $V(s)$ such that the ground state projection $P(s)$ of the Hamiltonian $H(s)$ is given by 
\begin{equation*}
P(s) = V(s)P(0)V^*(s).
\end{equation*}
These unitaries are the unique solution of
\begin{equation*}
-\mathrm i \frac{\dd}{\dd s} V(s) = K(s) V(s), \qquad U(0) = \I,
\end{equation*}
where the generator $K(s)$ can be written as
\begin{equation} \label{def: d}
K(s) = \int \dd t \; W(t) \tau_t^{H(s)}\Big(\frac{\dd}{\dd s} H(s)\Big).
\end{equation}
Here, $W \in L^\infty(\bbR)\cap L^1(\bbR)$ is a specific function~\cite{Sven} such that
\begin{equation}\label{eq: decay of W}
\left\vert W(t)\right\vert = \caO(t^{-\infty}),\qquad \left\vert \int_t^{\infty} \dd t' \; W(t') \right\vert= \caO(t^{-\infty}),
\end{equation}
as $t\to\infty$.
\end{lemma}

Here and below, we use the notation $\caO(t^{-\infty})$ for a function that decays to zero faster then any rational function.

Since this will be essential for the proofs, we note that by the Lieb-Robinson bound and the fast decay of $W$, the support of $K(s)$ defined in~(\ref{def: d}) is in a neighbourhood of the support of $\frac{\dd}{\dd s}H(s)$. As discussed in Section~\ref{sec:fluxes}, this is $\partial X_1 \cup \partial X_2$ is the case of the twist-antitwist Hamiltonian but only $\partial X_1^- \cup \partial X_2^-$ for the twist Hamiltonians, see Figure~\ref{fig: four squares} again. The support is in fact only $\partial X_j$, resp. $\partial X_j^-$, if $s\mapsto \phi(s)$ is chosen so that only $\phi_j(s)$ varies.


\section{Proofs}
 
\noindent The main point is to prove Propostion~\ref{prop: constant curvature}.
Using the quasi-adiabatic generators $\tilde{K}_j$ associated to changes of $\tilde{P}$ in directions $\phi_j$, namely
\begin{equation*}
\partial_j\tilde{P}(\phi) = \mathrm i[\tilde{K}_j(\phi), \tilde{P}(\phi)],\qquad \tilde{K}_j(\phi) = \int_{\mathbb R} W_{}(t)\tau_t^{\tilde H(\phi)}\left(\partial_j \tilde H(\phi)\right) \mathrm dt
\end{equation*}
and the cyclicity of the trace (the Hilbert space $\caH_\Gamma$ is finite dimensional), we have
\begin{equation}\label{Kubo}
\curv(\phi) = \iu \Tr(\tilde{P}(\phi)[\tilde{K}_1(\phi),\tilde{K}_2(\phi)]).
\end{equation}

We are going to show that this is asymptotically constant by comparing the expression inside the trace with such expression for $P(\phi), K(\phi)$ associated to the twist-antitwist Hamiltonian $H(\phi)$. Although for technical reasons, the proof below is phrased slightly differently, the heart of the argument can be presented in the following brief way. By~(\ref{GT},\ref{HTAT}), the family $H(\phi)$ is isospectral and hence 
\begin{equation}\label{eq: cov of P}
P(\phi) = e^{\mathrm i \langle \phi, Q\rangle} P e^{-\mathrm i \langle \phi, Q\rangle},
\end{equation}
where $\langle \phi, Q\rangle = \phi_1 Q_1 + \phi_2 Q_2$.
Furthermore, $[Q_1,Q_2]=0$ so that
\begin{equation}\label{indep of phi}
K_j(\phi) = e^{\mathrm i \langle \phi, Q\rangle} K_j e^{-\mathrm i \langle \phi, Q\rangle}
\end{equation}
as well, and hence
\begin{equation}\label{isosp}
P(\phi)[K_1(\phi), K_2(\phi)] = e^{\mathrm i \langle \phi, Q\rangle} P[K_1, K_2 ] e^{-\mathrm i \langle \phi, Q\rangle}.
\end{equation}
As discussed above, $K_j(\phi)$ are supported in a neighbourhood of both ribbons of $\partial X_j$ while $\tilde K_j(\phi)$ are supported in a neighbourhood of $\partial X_j^-$ only. In fact, more can be said: since $\partial_j H(\phi)$ is a sum of two terms with disjoint supports, $K_j(\phi)$ is itself a sum of two terms supported in a neighbourhood of $\partial X_j^-$ and $\partial X_j^+$ respectively. Hence the commutator $[K_1(\phi),K_2(\phi)]$ is a sum of four terms, each supported in a neighbourhood of a different corner. On the other hand, $[\tilde K_1(\phi),\tilde K_2(\phi)]$ is supported in a neighbourhood of the single corner $\Delta$ --- we shall take an $L/4$-fattening of $\Delta$ --- where it is approximately equal to the restriction of $[K_1(\phi),K_2(\phi)]$. Hence, 
\begin{align}
\curv(\phi) = \iu \Tr(\tilde{P}(\phi)[\tilde{K}_1(\phi),\tilde{K}_2(\phi)]) 
&= \iu \Tr\left(\tilde{P}(\phi)\tr_{\Gamma\setminus\Delta^{L/4}}[\tilde{K}_1(\phi),\tilde{K}_2(\phi)]\right) + \caO(L^{-\infty}) \nonumber\\
&= \iu \Tr\left(\tilde{P}(\phi)\tr_{\Gamma\setminus\Delta^{L/4}}[K_1(\phi),K_2(\phi)]\right) + \caO(L^{-\infty})\label{kappa}
\end{align}
where we noted in the second equality that the restriction of $[K_1(\phi),K_2(\phi)]$ to the corner $\Delta$ is equal to the partial trace over $\Gamma\setminus\Delta$ because each of the four terms is traceless. Now, in the neighbourhood of $\Delta$, the ground states $P(\phi)$ and $\tilde P(\phi)$ are approximately equal. Indeed, as noted in~(\ref{supp T vs TAT}) $H(\phi)-\tilde H(\phi)$ is supported away from $\Delta$ and local perturbations perturb gapped ground states locally, see~\cite{Sven}. Hence, (\ref{kappa}) can further be written as
\begin{equation*}
\curv(\phi) = \iu \Tr(P(\phi)\tr_{\Gamma\setminus\Delta^{L/4}}[K_1(\phi),K_2(\phi)]) + \caO(L^{-\infty})
\end{equation*}
By~(\ref{isosp}), the fact that the local charge is on-site and cyclicity again, this is independent of $\phi$, which concludes the argument. It is interesting to note that the corner $\Delta$ has an echo in the analysis of the non-interacting situation, see~\cite{AvronSeilerSimonCMP94, BulkEdgeMobility, KitaevHoneycomb}.

\subsection{The case of the fractional quantum Hall effect.}
The description of the simple mechanism of the proof above allows us to explain how the results of this paper can be extended to cover fractional conductance, as also explained in Section~9 of the original~\cite{HastingsMichalakis}. Let us modify Assumption \ref{tilde gap} by allowing that there is a $q$-dimensional spectral subspace of $\tilde{H}(\phi)$, the range of a spectral projector $\tilde P(\phi)$. It is not important that the Hamiltonian $\tilde{H}(\phi)$ is degenerate on this space, but we still call the range of $\tilde P(\phi)$ the ground state space and we require an $L$-independent gap to other parts of the spectrum.  By construction, $q$ is $\phi$-independent and we also assume it to be $L$-independent, for $L$ large enough. Additionally, we require a \emph{topological order} condition, see~(\ref{TQO}) below.

 The argument has two parts. First of all, let $\omega_{\phi} = q^{-1}\tilde P(\phi)$ be the chaotic ground state, i.e.\ the incoherent superposition of all ground states. The argument above runs unchanged but for a factor $q^{-1}$ that is carried through from the definition~(\ref{Kubo}). Hence  $\phi\mapsto\Tr(\tilde{P}(\phi)[\tilde{K}_1(\phi),\tilde{K}_2(\phi)])$ remains approximately constant and integrates to an integer, proving that the expression $\omega( [\tilde{K}_1,\tilde{K}_2]))$ (we suppress the $\phi$-dependence) is of the form $p/q$ for $p\in\bbN$. 

 Secondly, let $\tilde\omega=\tilde\omega_{\phi}$ be any (pure) ground state, i.e.\ a positive normalized functional that is supported on $\tilde P(\phi)$ and let us assume the topological order condition: for any local observable $O$ with support independent of $L$, we have 
\begin{equation}\label{TQO}
\tilde\omega (O) =  \omega (O) + \caO(L^{-\alpha})
\end{equation}
for some $\alpha>0$.
 Then,  since $[\tilde{K}_1,\tilde{K}_2]$ can be approximated by an observable located in the corner $\Delta^{L/4}$, the topological order condition implies that
\begin{equation*}
\tilde\omega ([\tilde{K}_1,\tilde{K}_2]) = \omega([\tilde{K}_1,\tilde{K}_2]) + \caO(L^{-\alpha})
\end{equation*}
This proves fractional quantization for any ground state.
Although the argument is compelling, one should keep in mind that there is to date no proven example of an interacting Hamiltonian exhibiting such fractional quantization with $p/q \notin\bbZ$.

\subsection{The actual proof}
In the following lemma, we compare $[\tilde K_1(\phi),\tilde K_2(\phi)]$ with ${G^\Delta}(\phi)$ (defined below), which is an adequate replacement of $\tr_{\Gamma\setminus\Omega}[K_1(\phi),K_2(\phi)]$. We denote by $\tau^{\phi}_t$ and $\tilde \tau^{\phi}_t$ the time-evolutions generated by $H(\phi)$ and $\tilde H(\phi)$. Let
\begin{equation*}
{G^\Delta}(\phi) := \int \dd t \; W_{}(t) \int \dd t' \; W_{}(t') \left[ \tau^{\phi}_t (\partial_1 H_{\Delta^{L/8}}(\phi)), \tau^{\phi}_{t'}(\partial_2 H_{\Delta^{L/8}}(\phi)) \right],
\end{equation*}
where we denote
\begin{equation*}
H_Z := \sum_{X\subset Z} \Phi(X)
\end{equation*}
for any subset $Z\subset\Gamma$. The following lemma establishes that ${G^\Delta}(\phi)$ is localized in a neighbourhood of $\Delta$ and that it is a good approximation of~$[\tilde K_1(\phi),\tilde K_2(\phi)]$.

\begin{lemma}\label{lma: K corner}
We have
\begin{equation*}
\left\Vert [\tilde K_1(\phi),\tilde K_2(\phi)] - \tr_{\Gamma \setminus \Delta^{L/4}} \left( {G^\Delta}(\phi) \right) \right\Vert = \caO(L^{-\infty})
\end{equation*}
and
\begin{equation*}
\tr_{\Gamma \setminus \Delta^{L/4}} \left( {G^\Delta}(\phi) \right) = U^*(\phi) \; \tr_{\Gamma \setminus \Delta^{L/4}} \left( {G^\Delta}(0) \right) \; U(\phi).
\end{equation*}
\end{lemma}
\begin{proof}
To prove the first estimate, we pick up a $\phi\in\bbT^2$ and drop the $\phi$ dependence in this proof for notational clarity. First of all, we note that the operator $G^{\Delta}$ is concentrated around the set $\Delta$. Indeed, the commutator $[\partial_1 H_{\Delta^{L/8}}, \partial_2 H_{\Delta^{L/8}}]$ is strictly supported on $\Delta^{L/8}$, the time evolution $\tau_t$ can be controlled using the Lieb-Robinson bound for short times, and the good decay properties of $W_{}$ take care of long times, see also~\cite{Sven}. To make this precise, we show that
\begin{equation}\label{eq:KDelta lives on Delta}
\Vert {G^\Delta} - \tr_{\Gamma \setminus \Delta^{L/4}} ( {G^\Delta} ) \Vert = \caO(L^{-\infty}).
\end{equation}
Using the good decay properties~(\ref{eq: decay of W}) of $W_{}$ we can restrict the integrals to $[-T, T]$ with $T = L/(32v)$,  making an error of order $\caO(L^{-\infty})$:
\begin{equation*}
\Vert {G^\Delta}- \tr_{\Gamma \setminus \Delta^{L/4}} ( {G^\Delta} ) \Vert \leq \left\Vert   \int_{-T}^T \dd t \; W_{}(t) \int_{-T}^T \dd t' \; W_{}(t')  f(t, t') \right\Vert +\caO(L^{-\infty}),
\end{equation*}
where the integrand is given by
\begin{equation*}
f(t, t') = \left( \I  - \tr_{\Gamma \setminus \Delta^{L/4}}  \right) [\tau^{}_t(\partial_1 H_{\Delta^{L/8}}), \tau^{}_{t'}(\partial_2 H_{\Delta^{L/8}})].
\end{equation*}
To estimate $\Vert f(t,t') \Vert$ we use Lemma~\ref{lma:trace of evolved} with $X = \Delta^{L/8}$, $r= L/8,\,\vert t\vert\leq L/32v$ and the fact that $\Vert  \partial_j H_{\Delta^{L/8}} \Vert  \leq CL$ to bound the integral by $\caO(L^{-\infty})$. The claim then follows.

The next step is to show that $[\tilde K_1, \tilde K_2]$ is close in norm to $G^{\Delta}$. Note first of all that, like in the argument above, we can restrict the integrals to $[-T, T]$ making an error of order $\caO(L^{-\infty})$. Therefore, it is sufficient to bound the norm of
\begin{equation}\label{eq: W tilde W}
\int_{-T}^T \dd t \; W_{}(t) \int _{-T}^T \dd t' \; W_{}(t')  \left(  [\tilde\tau^{}_t(\partial_1 \tilde H), \tilde\tau^{}_{t'}(\partial_2 \tilde H) ]   -    [\tau^{}_t(\partial_1 H_{\Delta^{L/8}}), \tau^{}_{t'}(\partial_2 H_{\Delta^{L/8}}) ]  \right).
\end{equation}
Since $\dist ( \Delta^{L/8}, \supp(H - \tilde H)  ) = 3L/8 - 2R$, it follows from Lemma~\ref{lma:Diff of H} that
\begin{equation*}
\Vert \tau^{}_t(\partial_j H_{\Delta^{L/8}}) - \tilde \tau^{}_{t}(\partial_j H_{\Delta^{L/8}}) \Vert = \caO(L^{-\infty})
\end{equation*}
for any $t \in [-T, T]$. Thus we can replace the evolutions $\tau^{}_t$ by $\tilde \tau^{}_t$, making an error that is again a $\caO(L^{-\infty})$. We can further replace $\partial_j H_{\Delta^{L/8}}$ by $\partial_1 \tilde H_{\Delta^{L/8}}$ without any error since by construction $H_{\Delta^{L/8}} = \tilde H_{\Delta^{L/8}}$ (as functions of $\phi$).
Altogether, we estimate the integrand of~(\ref{eq: W tilde W}) by 
\begin{align*}
\Vert W\Vert_\infty^2 &\Vert [\tilde\tau^{}_t(\partial_1 \tilde H), \tilde\tau^{}_{t'}(\partial_2 \tilde H) ]   -    [\tilde\tau^{}_t(\partial_1 \tilde H_{\Delta^{L/8}}), \tilde \tau^{}_{t'}(\partial_2 \tilde H_{\Delta^{L/8}}) ] \Vert \\
 &\leq \Vert W\Vert_\infty^2  \left(\Vert [\tilde \tau^{}_{t - t'}(\partial_1 \tilde H), \partial_2 \tilde H - \partial_2 \tilde H_{\Delta^{L/8}}]  \Vert + \Vert [\tilde \tau^{}_{t - t'}(\partial_1 \tilde H_{\Delta^{L/8}} - \partial_1 \tilde H), \partial_2 \tilde H_{\Delta^{L/8}}] \Vert \right)\\
& =\caO(L^{-\infty}),
\end{align*}
for any $t, t' \in [-T, T]$. To obtain the last estimate we used the Lieb-Robinson bound, noting that the supports of $\partial_1 \tilde H$ and $\partial_2 \tilde H - \partial_2 \tilde H_{\Delta^{L/8}}$ are separated by a distance $L/8$ while $v \vert t-t'\vert\leq L/16$. Hence,
\begin{equation*}
\left\Vert[\tilde K_1, \tilde K_2] - G^{\Delta}\right\Vert = \caO(L^{-\infty}),
\end{equation*}
which, together with~(\ref{eq:KDelta lives on Delta}), concludes the proof of the first claim. 

To get the covariance we note that 
$G^{\Delta}(\phi) = U^*(\phi) G^{\Delta}(0) U(\phi)$ which follows directly from $H(\phi) = U^*(\phi) H U(\phi)$. 
Then, the covariance of $\tr_{\Gamma \setminus \Delta^{L/4}}  G^{\Delta}(\phi)$ follows upon noting that $U(\phi)$ is a product over single site unitaries.
\end{proof}

\begin{proof}[Proof of Proposition~\ref{prop: constant curvature}] 
By~(\ref{Kubo}) and Lemma~\ref{lma: K corner}, 
\begin{equation*}
\curv(\phi) = \Tr \left( \tilde P(\phi) \; \tr_{\Gamma \setminus \Delta^{L/4}} \left( {G^\Delta}(\phi) \right) \right) + \caO(L^{-\infty}).
\end{equation*}
Recalling the general flux Hamiltonians defined in Section~\ref{sec:fluxes} and the relations (\ref{HT}\ref{HTAT}), Assumption~\ref{tilde gap} ensures that the spectral gap above the ground state energy does not close along the smooth interpolation $[0,1]\ni s \mapsto H(\phi, (s-1)\phi)$. Let $V(s)$ be the quasi-adiabatic unitaries corresponding to this homotopy, as provided by Lemma \ref{lem: quasi adiabatic}. By construction, the derivative $\frac{\dd}{\dd s} H(\phi, (s-1)\phi)$
is supported on $\partial X_1^+ \cup \partial X_2^+$. For any observable~$O$, we can write
$$
V^*(1)O V(1) = O+ \iu \int_0^1 V^*(s)  [K(s),O]    V(s) \dd s.
$$
Let us now assume that $O$ is supported in $\Delta^{L/4}$. Then by the expression \eqref{def: d} of $K(s)$, the fact the distance of the support of $\frac{\dd}{\dd s} H(\phi, (s-1)\phi)$ to $\Delta^{L/4}$ is $3L/4 - 2R$, the decay of $W(t)$ and the Lieb-Robinson bound, we deduce that 
\begin{equation*}
\Vert [K(s),O] \Vert  = \caO(L^{-\infty}).
\end{equation*}
Therefore, applying this with $O=\tr_{\Gamma \setminus \Delta^{L/4}} {G^\Delta}(\phi)$ and using $\tilde P(\phi)=V(1) P(\phi)  V^*(1)$ and cyclicity of the trace, we get 
\begin{equation*}
\Tr \left( \tilde P(\phi) \; \tr_{\Gamma \setminus \Delta^{L/4}} \left( {G^\Delta}(\phi) \right) \right) = \Tr\left( P(\phi) \;  \tr_{\Gamma \setminus \Delta^{L/4}} \left( {G^\Delta}(\phi) \right) \right) + \caO(L^{-\infty}).
\end{equation*}
Now, the covariance of $P(\phi)$ and of $\tr_{\Gamma \setminus \Delta^{L/4}} \left( {G^\Delta}(\phi) \right)$ provided in~(\ref{eq: cov of P}) and Lemma~\ref{lma: K corner} respectively, show that the trace on the right is in fact independent of $\phi$ by cyclicity, settling the claim. 
\end{proof}
\begin{proof}[Proof of Theorem \ref{thm:mainupgrade}]
Recall from~(\ref{Kubo}) that
\begin{equation}
\curv_L = \iu \omega_{0,L}([\tilde{K}_1,\tilde{K}_2]).
\end{equation}
where both the $\tilde K_j$'s and the state $\omega_{0,L}$ depend on $L$ (we have made the dependence explicit in the notation). The claim will follow from the fact that $[\tilde{K}_1,\tilde{K}_2]$ can be approximated uniformly in $L$ by a local observable $\Upsilon^\ell$ supported in $\Delta^\ell$. The error decays rapidly in $\ell$, and the expectation value of $\Upsilon^\ell$ converges by assumption.

Indeed, (\ref{def: d}) and the arguments repeatedly used in this article yield that $\tilde K_j$ is a sum of local terms in the form $\tilde K_j = \sum_{X\subset \Gamma}\Psi^\Gamma(X)$, where $\Vert \Psi^\Gamma(X) \Vert = \caO(\mathrm{diam} (X)^{-\infty})$ uniformly in $L$, and that $\Psi^\Gamma(X)$ converges in norm as $L\to\infty$ for any fixed $X$, see~\cite{Sven}. Hence, for any $\ell\in\bbN$, the local observable $\tr_{\Gamma\setminus \Delta^\ell}[\tilde{K}_1,\tilde{K}_2]$ converges to a $\Upsilon^\ell\in\caB(\caH_{\Delta^\ell})$ as $L\to\infty$. Since moreover,
\begin{equation*}
(\I - \tr_{\Gamma\setminus \Delta^\ell})([\tilde{K}_1,\tilde{K}_2]) = \caO(\ell^{-\infty}),
\end{equation*}
uniformly in $L$, we have
\begin{align*}
[\tilde{K}_1,\tilde{K}_2] - \Upsilon^\ell
&= (\I - \tr_{\Gamma\setminus \Delta^\ell})([\tilde{K}_1,\tilde{K}_2])
+ \left( \tr_{\Gamma\setminus \Delta^\ell} [\tilde{K}_1,\tilde{K}_2] - \Upsilon^\ell \right) \\
&= \caO(\ell^{-\infty}) + \caO(L^{-\infty}).
\end{align*}
Hence, 
\begin{equation*}
\lim_{L\to\infty }\left\vert \curv_L - \iu \omega_{0,L}(\Upsilon^\ell)\right\vert 
=\caO(\ell^{-\infty}).
\end{equation*}
By assumption $\omega_{0,L}(\Upsilon^\ell)$ converges, concluding the proof.
\end{proof}

{
\section{Acknowledgements}
\noindent We would like to thank Y.~Avron for his careful reading of the first version of this manuscript, and for his many comments which helped improve this article.
 WDR acknowledges the support of the Flemish Research Fund FWO under grant G076216N. AB, MF and WDR have been supported by the InterUniversity Attraction Pole phase VII/18 dynamics, geometry and statistical physics of the Belgian Science Policy.
}

\end{document}